\documentclass[11pt]{article}

\usepackage{eulervm}
\usepackage{setspace}
\usepackage[pdftex]{graphics}
\usepackage{graphicx}
\usepackage{color}
\usepackage{latexsym}
\usepackage{amsfonts}
\usepackage{amssymb, amsmath}
\usepackage{amsthm}
\usepackage{subcaption}
\usepackage{float}
\usepackage{bm}
\DeclareGraphicsExtensions{.pdf,.png,.jpg}

\usepackage{natbib} 
\usepackage{enumerate}

\usepackage[colorlinks=true, pdfstartview=FitV, linkcolor=blue, citecolor=blue, urlcolor=blue]{hyperref}

\usepackage{tikz}

\newtheorem{proposition}{Proposition}

\begin{document}

\title{RIF Regression via Sensitivity Curves
	}

\author{
		Javier Alejo\thanks{IECON-Universidad de la Rep\'ublica, Gonzalo Ram\'irez 1926, C.P. 11200, Montevideo, Uruguay. E-mail: \texttt{javier.alejo@fcea.edu.uy}}\\
		Gabriel Montes-Rojas\thanks{IIEP-BAIRES - Universidad de Buenos Aires and CONICET.  Email: \texttt{gabriel.montes@fce.uba.ar}}\\
		Walter Sosa-Escudero\thanks{Universidad de San Andr\'es and CONICET, Buenos Aires, Argentina. Email: \texttt{wsosa@udesa.edu.ar }}
	}

\maketitle

\begin{abstract}
This paper proposes an empirical method to implement the recentered influence function (RIF) regression of Firpo, Fortin and Lemieux (2009), a relevant method to study the effect of covariates on many statistics beyond the mean. In empirically relevant situations where the influence function is not available or difficult to compute, we suggest to use the \emph{sensitivity curve} (Tukey, 1977) as a feasible alternative. This may be computationally cumbersome when the sample size is large. The relevance of the proposed strategy derives from the fact that, under general conditions, the sensitivity curve converges in probability to the influence function. In order to save computational time we propose to use a cubic splines non-parametric method for a random subsample and then to interpolate to the rest of the cases where it was not computed. Monte Carlo simulations show good finite sample properties. We illustrate the proposed estimator with an application to the polarization index of Duclos, Esteban and Ray (2004).
\end{abstract}


\pagenumbering{arabic}

\baselineskip15pt

\newpage

\maketitle

\section{Introduction}

The \emph{recentered influence function} (RIF) regression, as proposed by Firpo, Fortin and Lemieux (2009), is a powerful tool to study the impact of changes in covariates on the unconditional distribution of a given outcome variable. Let $Y$ be a random variable with cumulative distribution function $F$, and $v(F)$ any `functional' of interest related to $F$.  For example, if $Y$ is income, $v(F)$ can be the mean, the Gini index, a quantile, or the poverty rate. The RIF is defined as $RIF(y,v,F)=v(F)+IF(y,v,F)$, where $IF(y,v, F)$ is the \emph{influence function} (IF) (Hampel, 1974) that measures the marginal impact of a particular data point in the support of $F$ in the value of $v(F)$. Influence functions play a key role in the robust statistics literature. 

Firpo et al. (2009, 2018) note that since $E[RIF(Y,v,F)]=v(F)$, by the law of iterated expectations $E_{X} \left[E_{Y|X} RIF(Y,v,F) \right]=v(F)$, and show that the effect on $v(F)$ that arises from shifting a scalar covariate from $X$ to $X+t$, where $t\downarrow 0$, is given by: 

\[\int \frac{d E[RIF(Y,v,F)|X=x]}{dx} dF(x).\] 
Hence, by properly modelling $E[RIF(Y,v,F)|X=x]$ in a regression fashion, the effect of $X$ on $v$ can be recovered as an `average derivative' of regressing $RIF(Y,v,F)$ on $X$. The implementation of the method requires to construct $RIF(Y,v,F)$ analytically for the functional of interest $v$ and then to regress it on $X$. In many relevant cases the IF required to obtain $RIF(Y,v,F)$ is immediately available; Fortin, Lemieux and Firpo (2011) present a useful `catalog' that includes the mean, the quantiles, the variance and the Gini index (see also Essama-Nssah and Lambert (2015) and Cowell and Flachaire (2015)).  However, there are many examples where this is not the case. Our paper proposes an alternative in these situations.

In this paper we propose a practical computation method based on the \emph{sensitivity curve} (SC) (Tukey, 1977). This procedure consists in comparing the full sample functional $v$ with that computed when the $j-$th observation is left out; this is the influence of this particular observation on the empirical version of $v$.
The relevance of the proposed strategy derives from the fact that, under general conditions, the SC converges in probability to the IF (see Nasser and Alam (2006) for a discussion). We provide an intuitive proof of this result.

The SC has some practical advantages over the IF. First, even when analytically available, in many cases the estimation of the IF involves dealing with the problem of selection of the meta-parameters, like bandwidths, which may add further complications. Second, in some relevant cases the IF may be difficult when not impossible to derive analytically. As an example of this case we study the Duclos, Esteban, and Ray (2004) polarization index, where for the general case there is no analytical functional form of the IF (see Appendix A2 for a summary of the construction and motivation of this index). Finally, many relevant examples where the IF can be easily derived involve additive or quasi-additive measures that do not apply to many important situations. 

This paper is organized as follows. Section \ref{IF} presents the main statistical derivations. Section \ref{splines} discusses the cubic spline method to interpolate the SC and considerably reduce computation time. Section \ref{MC} provides finite sample Monte Carlo simulations. Section \ref{Empirical} discusses an empirical exercise that shows that the performance of the SC is close to that of the analytical IF. 

\section{Influence via sensitivity curves} \label{IF}

Let $v(F)$ be a real-valued functional, where $v: \mathcal{F}_v \rightarrow \mathbb{R}$ and $\mathcal{F}_v$ is a class of distribution functions such that $F \in \mathcal{F}_v$ if $|v(F)|<\infty$. Consider two cumulative distribution functions (CDFs), $F$ and $G$, and let $H_{t, F, G} = t G+(1-t) F$, $t\in[0,1]$. Then, using the Von Mises (1947) expansion:

\begin{equation}\label{v}
v(H)=v(F)+t \partial v\left(H_{t, F, G}\right) / \partial t \left.\right|_{t=0}+r(t, F, G),
\end{equation}
with
\begin{equation}
\begin{aligned} \partial v\left(H_{t, F, G}\right) / \partial\left.t\right|_{t=0} &=\lim _{t \downarrow 0} \frac{v\left(H_{t, F, G}\right)-v(F)}{t} \\ &=\int \psi(y) d(G-F)(y). \end{aligned}
\end{equation}
When  $G = \Delta_y$ and $\Delta_y$ is the CDF of a random variable with probability mass of 1 at $y$, $\psi(y) = \partial v\left(H_{t, F, \Delta_{y}}\right) / \partial\left.t\right|_{t=0}$ is the \emph{influence function (IF)} of the functional $v$, labeled as $IF(y,v,F)$ (see Huber and Ronchetti (2009) for a general discussion; here we are following the derivation in Firpo et al. (2009, p.956)).

Consider now the last term in eq. \eqref{v}. Following Von Mises (1947):
\begin{equation}\label{r}
r(t, F, G)=\frac{\tilde{t}^{2}}{2} \partial^{2} v\left(H_{t, F, \Delta_{y}}\right) / \partial\left.t^{2}\right|_{t=0},
\end{equation}
for some $\tilde{t} \in[0, t]$, where
\begin{equation}
\partial^{2} v\left(H_{t, F, G}\right) / \partial t^{2}\left.\right|_{t=0}=\iint \phi(y, z) d(G-F)(y) d(G-F)(z),
\end{equation}
with $\phi(y,z)$ a symmetric function; again, see Von Mises (1947, p. 325) for details. Note that if $v(F)=v(cF)$ for all $c>0$ (scale invariance) then:

\begin{equation*}
\begin{array}{ll}{\text { (i) }} & {\int \psi(y) d F(y)=0} \\ {\text { (ii) }} & {\iint \phi(y, z) d F(y) d F(z)=0}\end{array}
\end{equation*}
The proof of (i) and (ii) follows from Jaeckel (1972).\footnote{Let $G=2F$, then $H_{(t,F,G)}=(1+t)F=cF$ and then by the invariance to scale 
	\begin{equation*}
\partial v\left(H_{t, F, G}\right) / \partial\left.t\right|_{t=0}=\lim _{t \downarrow 0} \frac{v(c F)-v(F)}{t}=\lim_{t \downarrow 0} \frac{0}{t}=0.
\end{equation*}
Moreover,	
\begin{equation*}
\partial^{2} v\left(H_{t, F, \Delta_{y}}\right) / \partial\left.t^{2}\right|_{t=0}=0.
\end{equation*}
}

The \emph{recentered influence function (RIF)}, is defined as $RIF(y,v,F)\equiv v(F)+IF(y,v,F)$, where, trivially, $E[IF(y,v,F)] = v(F)$, from property (i) above. Firpo et al. (2009) develop a RIF-regression framework that is similar to a standard regression except that the dependent variable, $Y$, is replaced by the IF of the statistic of interest, which allows to  estimate the effects of covariates $X$ on $v(F)$. 

Unfortunately, not all indicators have an IF with a specific analytical form and thus the RIF-regression may not be practically feasible. Our proposal consists of replacing the IF by the SC.

Let $\{y_i\}_{i=1}^n$ be an $iid$ sample and define $v_n=v(F_n)$ as the sample counterpart of $v(F)$, and let $v_n^{(j)}=v(F_n^{(j)})$ denote the case where $j-$th observation is left out, then:
\begin{eqnarray*}
F_{n}(y)&=&\frac{1}{n} \sum_{i=1}^{n} \; 1\left(y_{i} \leq y\right) \\
F_{n}^{(j)}(y)&=&\frac{1}{n-1} \sum_{i \neq j} 1\left(y_{i} \leq y\right).
\end{eqnarray*}
The \emph{sensitivity curve (SC)} is defined as
\begin{equation}
SC\left(y_{j}, v_{n}, F_{n}\right) \equiv n \cdot\left[v\left(F_{n}\right)-v\left(F_{n}^{(j)}\right)\right].
\end{equation}

The key property that links the IF to the SC is the following:

\begin{proposition}
Assume that $v(F)$ is twice continuously differentiable with respect to $F$ and $\psi(y)$ and $\phi(y,z)$ exist, and that $v$ is invariant to scale (i.e., $v(F)=v(cF)$ for $c>0$).
Then, $SC\left(y_{j}, v_{n}, F_{n}\right) \stackrel{p}{\rightarrow} I F\left(y_{j}, v, F\right)$ as $n\rightarrow\infty$.
\end{proposition}
\begin{proof}
	See the Appendix A1.
\end{proof}

Consequently, if the functional $v$ is smooth enough, the SC can be used instead of the analytical IF. Nasser and Alam (2006) show that Fr\'echet differentiability is sufficient for consistency. Of course smoothness may be considered a strong requirement. For example, for the case of quantiles $IF(y, Q_\tau,F) =(y-1[y\leq Q_\tau(F)])/f_y(Q_\tau(F))$, where $1[.]$ is an indicator function, $f_y(Q_\tau(F))$ is the density of the marginal distribution of $y$ evaluated at the $\tau$-quantile, and $Q_\tau(F)$ is the population $\tau$-quantile of the unconditional distribution of $y$. The indicator function  makes it non twice differentiable.

The \emph{recentered sensitivity curve (RSC)} is defined as:
\[RSC(y_j,v_n,F_n) \equiv v_n+SC(y_j,v_n,F_n)\]
Trivially $RSC\stackrel{p}{\rightarrow}RIF$. Hence, our proposal is to replace RIF with RSC. An then, to apply regression models to approximate distributional effects as in Firpo et al. (2009) RIF regression method.

\section{Computation of the SC via cubic splines}\label{splines}

The RSC method described above requires to compute the functional $v$ $n+1$ times, that is, the original using the entire sample plus all of the leave-one-out cases (i.e. $n$). This may be computationally cumbersome when the sample size is large.  In order to save computational time we propose to compute the $SC$ for a random sub-sample and then to interpolate to the rest of the domain of $y$. A widely used method to perform this type of adjustment is splines since they are implemented through a flexible functional form that is linear in parameters. In particular we use the restricted cubic splines method to interpolate for the value of the RSC for the cases where it was not computed. 

Cubic splines are piecewise-polynomial line segments whose function values and first and second derivatives agree at the boundaries where they join. The boundaries of these segments are called knots, and the fitted curve is continuous and smooth at the knot boundaries (see Smith, 1979;  Wegman and Wright, 1983; Harrell, 2001, ch. 2). Let $k_j$, $i = 1, . . . , K$, be the knot values defined in the support of $y$, then the equation of the cubic spline is

$$ J(y)= \beta_0 + \beta_1 y + \beta_2 y^2 + \beta_3 y^3 + \sum_{j=1}^{K}\gamma_j(y-k_j)_{+}^3,$$
where $u_{+}:=max(u,0)$. A common problem with cubic spline is that it fits poorly in the tails. One way to deal with this is by restricting $J(y)$ to be linear for  $y<k_1$ and $y>k_K$. This requirement is satisfied when $\beta_2=\beta_3=0$, $\sum_{j=1}^{K}\gamma_j=0$ and $\sum_{j=1}^{K}\gamma_j k_j=0$. Replacing this in the $J(y)$ equation, Durrleman and Simon (1989) show that the restricted cubic spline is then
$$ J_R(y)= \beta_0 + \beta_1 y + \sum_{j=1}^{K-2}\gamma_j h_j(y),$$
where 
$$ h_j(y) = (y-k_j)_{+}^3 + \frac{k_{K}-k_{j}}{k_{K}-k_{K-1}}(y-k_{K-1})_{+}^3 + \frac{k_{K-1}-k_{j}}{k_{K}-k_{K-1}}(y-k_{j})_{+}^3$$
for $j=1,..,K-2$. Note that $J_R(y)$ is linear in parameters and therefore can be estimated by ordinary least-squares (OLS) methods using $y$ and the $\{h_j\}_{j=1}^{K-2}$ auxiliary variables.

The interpolation of the RSC function proceeds in three steps.

In the first step, we select the knots on the full sample of $\{y_i\}_{i=1}^n$ and create the auxiliar variables $\{h_{1i},...,h_{(K-2)i}\}_{i=1}^n$ that corresponds to the restricted cubic spline method. 

In the second step, we consider a random sample without replacement of $\{y_i,h_{1i},...,h_{(K-2)i}\}_{i=1}^n$ denoted by $\{y_{i^*},h_{1i^*},...,h_{(K-2)i^*}\}_{i^*=1}^{n^*}$, where $n^*<n$. For this random sample we compute the RSC for each of the $n^*$ observations, say $\{\widetilde{SC}_{i^*}\}_{i^*=1}^{n^*}$. This is the step that significantly reduces the computation time (see the simulations in the Monte Carlo section). Moreover, we estimate the parameters $(\beta_0,\beta_1,\gamma_1,...,\gamma_{K-2})$ by fitting an OLS regression of $\widetilde{SC}_{i^*}$ as a function of $y_{i^*},h_{1i^*},...,h_{(K-2)i}^*$.

Finally, in the third step, we apply the estimated linear regression coefficients  to compute the cubic spline interpolation for the full sample,  
$$\widehat{SC}_{i}=\hat\beta_0 + \hat\beta_1 y_{i} +\hat\gamma_1 h_{1i}+...+\hat\gamma_{K-2} h_{(K-2)i}.$$
The $\widehat{SC}_{i}$ interpolated values are then used in the RSC method to compute the effect of covariates on the given functional $v(F)$ (see Orsini and Greenland, 2011, and Newson, 2012, for a discussion of how this interpolation works).

\section{Monte Carlo experiments}\label{MC}

In this section we run some numerical simulation exercises to evaluate the computational and statistical performance of the proposed method. Throughout this section we use the following baseline model,

\begin{equation*}
	Y = 20 + X + W,
\end{equation*}
where $X \sim Uniform(0,1)$ is the observable covariate and $W$ is the unobservable variable. Then we use the two alternative models:

\begin{enumerate}
	\item Location-scale model: $W = (1+X)U$, with $U \sim N(0,1)$. 
	\item Location-bimodal model: $W = (D(-4+U(2-X))+(1-D)(4+Z(2-X)))/5$, with $(U,Z,D)$ independent and with distributions $U \sim N(0,1)$, $Z \sim N(0,1)$ and $D$ Bernoulli with $Pr(D=1)=0.50$.
\end{enumerate}

In all the exercises in this section we use \texttt{STATA} version 14.1 MP (64-bit) installed on a computer with 16 GB of RAM, an Intel Core i7 processor and Windows 10 operating system.

\subsection{Finite sample performance}

We use 1000 Monte Carlo simulations to evaluate the estimators' performance and compute Bias, Variance and MSE (mean-squared error). We consider two sample sizes of $n=500$ and $n=5000$. To compute the population parameter we use the DGPs with 10 million observations and where we compute  the numerical derivative of a change $x' = x + \epsilon$, that is, $\frac{v(F)-v(F')}{\epsilon}$ where $F$ and $F'$ are the induced distribution functions of the corresponding DGP with $X$ and $X+\epsilon$, respectively, and with $\epsilon=0.0001$.

We evaluate 3 different functionals: variance (Table \ref{table:MCvar}), Gini coefficient (Table \ref{table:MCgini}) and DER polarization index with $\alpha=0.5$ (Tables \ref{table:MCder} and \ref{table:MCdersp}). For the variance and Gini we have an analytical formula of the IF. As such we compute the RIF effect together with the RSC proposed method. For the DER polarization we can only report the RSC effect (see the Appendix A2 for a description of this index). In all cases, for the RSC computation we report the full sample RSC method and the splines approximation (RSC(sp)). For the latter we use 100 points for $n=500$ and 1000 for $n=5000$.

Table \ref{table:MCvar} shows the performance of the proposed method for computing the marginal effect on the variance. The simulations show that the proposed RSC method has a similar performance to that of RIF, which is close to the population parameter in terms of Bias and MSE. The spline approximation has a weaker performance of $n=500$ but it is similar to the full sample RSC for $n=5000$. 

\begin{table}
	\centering
	\caption{Variance}\label{table:MCvar}
	\begin{tabular}{c ccc| ccc}
\hline													
&	\multicolumn{3}{c}{$n=500$}					&	\multicolumn{3}{c}{$n=5000$}					\\
\hline													
&	RIF	&	RSC	&	RSC(sp)	&	RIF	&	RSC	&	RSC(sp)	\\
\hline													
\multicolumn{7}{c}{(i): location-scale model}													\\
\hline													
Population	&	3.003	&	3.003	&	3.003	&	3.003	&	3.003	&	3.003	\\
Mean	&	2.999	&	3.017	&	3.189	&	2.995	&	2.997	&	2.981	\\
Bias	&	-0.003	&	0.015	&	0.187	&	-0.008	&	-0.006	&	-0.021	\\
Var	&	0.369	&	0.374	&	0.460	&	0.035	&	0.035	&	0.035	\\
MSE	&	0.369	&	0.374	&	0.494	&	0.035	&	0.035	&	0.036	\\
\hline													
\multicolumn{7}{c}{(ii): location-bimodal model}													\\
\hline													
Population	&	-0.120	&	-0.120	&	-0.120	&	-0.120	&	-0.120	&	-0.120	\\
Mean	&	-0.117	&	-0.118	&	-0.224	&	-0.119	&	-0.120	&	-0.131	\\
Bias	&	0.003	&	0.002	&	-0.104	&	0.000	&	0.000	&	-0.011	\\
Var	&	0.012	&	0.012	&	0.014	&	0.001	&	0.001	&	0.001	\\
MSE	&	0.012	&	0.012	&	0.025	&	0.001	&	0.001	&	0.001	\\
\hline													
	\end{tabular}

\footnotesize{
Note: own calculations using 1000  Monte  Carlo  simulations.
}
	
\end{table}

Table \ref{table:MCgini} shows the performance of the proposed method for computing the marginal effect on the Gini coefficient. The results are in line with those for the variance: the RSC has a good performance relative to the RIF. For this case, however, the RSP(sp) approximation is much closer to the full sample RSC, and as such there is a minimum loss in efficiency for using the Spline interpolation.

\begin{table}
	\centering
	\caption{Gini}\label{table:MCgini}
	\begin{tabular}{c ccc| ccc}
\hline													
&	\multicolumn{3}{c}{$n=500$}					&	\multicolumn{3}{c}{$n=5000$}					\\
\hline													
&	RIF	&	RSC	&	RSC(sp)	&	RIF	&	RSC	&	RSC(sp)	\\
\hline													
\multicolumn{7}{c}{(i): location-scale model}													\\
\hline													
Population	&	2.473	&	2.473	&	2.473	&	2.473	&	2.473	&	2.473	\\
Mean	&	2.517	&	2.521	&	2.430	&	2.512	&	2.512	&	2.550	\\
Bias	&	0.044	&	0.048	&	-0.043	&	0.039	&	0.039	&	0.077	\\
Var	&	0.252	&	0.254	&	0.277	&	0.023	&	0.023	&	0.024	\\
MSE	&	0.254	&	0.257	&	0.278	&	0.025	&	0.025	&	0.030	\\
\hline													
\multicolumn{7}{c}{(ii): location-bimodal model}													\\
\hline													
Population	&	-0.322	&	-0.322	&	-0.322	&	-0.322	&	-0.322	&	-0.322	\\
Mean	&	-0.355	&	-0.358	&	-0.479	&	-0.359	&	-0.360	&	-0.342	\\
Bias	&	-0.033	&	-0.035	&	-0.156	&	-0.037	&	-0.037	&	-0.020	\\
Var	&	0.031	&	0.031	&	0.036	&	0.003	&	0.003	&	0.003	\\
MSE	&	0.032	&	0.032	&	0.061	&	0.004	&	0.004	&	0.003	\\
\hline													
	\end{tabular}

\footnotesize{
Note: own calculations using 1000  Monte  Carlo  simulations.
}

\end{table}

Finally we consider the analysis of the DER polarization index with $\alpha=0.5$. As discussed above there is no analytical IF for this model, and therefore the use of the RSC is the only alternative to evaluate the effect of the covariates on the DER index.  
For this case, we also evaluate alternative models for the RSC regression models. Given that we will not be able to derive the functional form of the conditional model of the IF conditional on $X$, we compute three different alternatives: linear ($a+bX$), quadratic  ($a+bX+cX^2$) and cubic polynomials ($a+bX+cX^2+dX^3$). Then we compute the average partial effects, that is, $b$ for the linear case, $b+2c\bar X$ for quadratic and  $b+2c\bar X+3d \bar{X^2}$ for the cubic polynomial case. 

Table \ref{table:MCder} shows the simulation results for the full-sample RSC computation and Table \ref{table:MCdersp} for the Spline interpolation. The location-scale model works similarly across methods, with large reduction in Bias and MSE when the largest sample size is used. The location-bimodal model, however, shows considerable heterogeneity across models. In most cases the quadratic approximation seems to correctly capture the effect of a marginal effect of $X$ on the DER index. For the Splines, the sample size requirement seems to be more demanding than in previous models.

\begin{table}
	\centering
	\caption{DER($\alpha=0.5$) - RSC}\label{table:MCder}
	\begin{tabular}{c ccc| ccc}
		\hline													
		&	\multicolumn{3}{c}{$n=500$}					&	\multicolumn{3}{c}{$n=5000$}					\\
		\hline													
		&	Linear	&	Quadratic	&	Cubic	&	Linear	&	Quadratic	&	Cubic	\\
		\hline													
		\multicolumn{7}{c}{(i): location-scale model}													\\
		\hline													
		Population	&	2.153	&	2.153	&	2.153	&	2.153	&	2.153	&	2.153	\\
		Mean	&	2.249	&	2.249	&	2.193	&	2.226	&	2.226	&	2.146	\\
		Bias	&	0.096	&	0.096	&	0.040	&	0.073	&	0.073	&	-0.007	\\
		Var	&	0.217	&	0.219	&	0.687	&	0.020	&	0.019	&	0.066	\\
		MSE	&	0.226	&	0.228	&	0.689	&	0.025	&	0.025	&	0.066	\\
		\hline													
		\multicolumn{7}{c}{(ii): location-bimodal model}													\\
		\hline													
		Population	&	0.291	&	0.291	&	0.291	&	0.291	&	0.291	&	0.291	\\
		Mean	&	0.441	&	0.443	&	0.106	&	0.657	&	0.657	&	0.229	\\
		Bias	&	0.149	&	0.151	&	-0.186	&	0.366	&	0.366	&	-0.063	\\
		Var	&	0.190	&	0.186	&	0.605	&	0.024	&	0.023	&	0.082	\\
		MSE	&	0.212	&	0.209	&	0.640	&	0.158	&	0.157	&	0.086	\\
		\hline													
											
	\end{tabular}

\footnotesize{
Note: own calculations using 1000  Monte  Carlo  simulations.
}

\end{table}

\begin{table}
	\centering
	\caption{DER($\alpha=0.5$) - RSC (sp)}\label{table:MCdersp}
	\begin{tabular}{c ccc| ccc}
		\hline													
		&	\multicolumn{3}{c}{$n=500$}					&	\multicolumn{3}{c}{$n=5000$}					\\
		\hline													
		&	Linear	&	Quadratic	&	Cubic	&	Linear	&	Quadratic	&	Cubic	\\
		\hline													
		\multicolumn{7}{c}{(i): location-scale model}													\\
		\hline													
		Population	&	2.153	&	2.153	&	2.153	&	2.153	&	2.153	&	2.153	\\
		Mean	&	2.004	&	2.004	&	1.970	&	2.271	&	2.271	&	2.202	\\
		Bias	&	-0.149	&	-0.149	&	-0.183	&	0.118	&	0.119	&	0.049	\\
		Var	&	0.270	&	0.271	&	0.798	&	0.019	&	0.019	&	0.067	\\
		MSE	&	0.292	&	0.293	&	0.832	&	0.033	&	0.033	&	0.069	\\
		\hline													
		\multicolumn{7}{c}{(ii): location-bimodal model}													\\
		\hline													
		Population	&	0.291	&	0.291	&	0.291	&	0.291	&	0.291	&	0.291	\\
		Mean	&	-0.971	&	-0.970	&	-1.152	&	0.343	&	0.343	&	0.084	\\
		Bias	&	-1.263	&	-1.262	&	-1.443	&	0.052	&	0.052	&	-0.208	\\
		Var	&	0.162	&	0.161	&	0.484	&	0.027	&	0.026	&	0.076	\\
		MSE	&	1.756	&	1.753	&	2.568	&	0.029	&	0.029	&	0.119	\\
		\hline													
		
	\end{tabular}

\footnotesize{
Note: own calculations using 1000  Monte  Carlo  simulations.
}

\end{table}

\subsection{Computing time}

We analyze the goodness of fit of the spline interpolation by simulating a random realization of $n=1000$ using the location-scale model. Figure \ref{fig:figure-splines} shows the RSC computed with the complete sample together with the spline interpolation RSC(sp) using a random 10\% of the original sample. Although the RSC of the DER(0.5) seems to be quite complicated to approximate compared to those of Gini and variance, the adjustment of the spline seems to be reasonable for the three indicators analyzed.

\begin{figure}
\caption{Comparison of RSC and RSC(sp) fit}
\label{fig:figure-splines}
\centering
\includegraphics[scale=0.22]{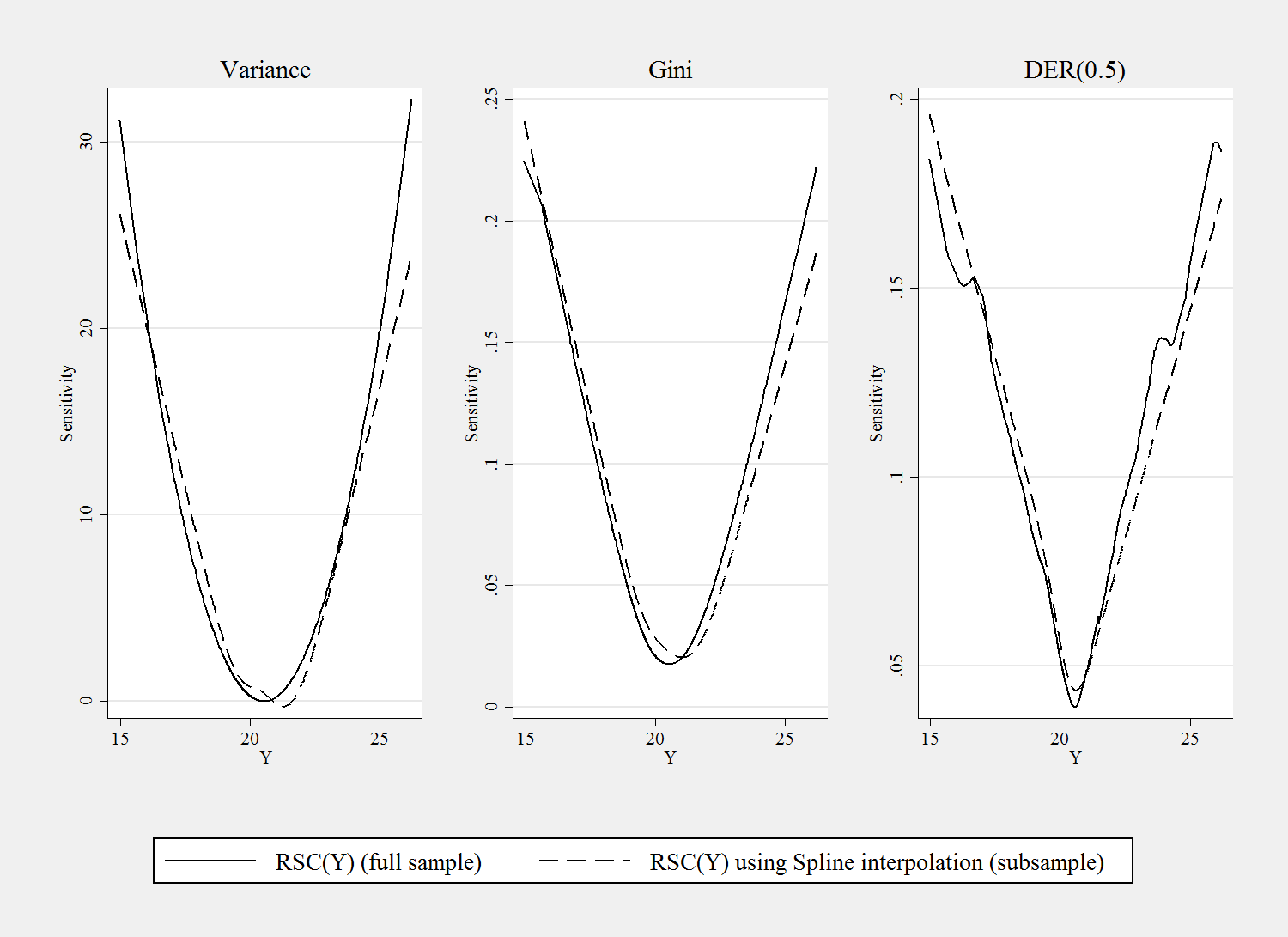}

\footnotesize{
Notes: own calculations using a random draw with sample size $n=1000$, $n^*=100$, \texttt{STATA} 14.1 MP (64-bit), 16 GB of RAM, an Intel Core i7 processor and Windows 10 operating system.
}

\end{figure}

Table \ref{table:PCtimes} shows the average computation time of the RSC and RSC(sp) with different sample sizes. For this exercise we use 50 random samples generated with the location-scale model.
In all cases, a subsample of 10\% of the original sample was considered for the RSC(sp) interpolation.

\begin{table}
	\caption{Average computing time (seconds)}\label{table:PCtimes}

\centering

\begin{tabular}{ccccc}
\hline
{\bf Index} & {\bf Sample} &  {\bf RSC} & {\bf RSC(sp)} & {\bf RSC(sp)/RSC} \\
\hline
           & $500\leq n \leq 1100$ &       0.17 &       0.03 &     19.7\% \\

  Variance & $1100< n \leq 1800$ &       0.43 &       0.04 &      9.4\% \\

           & $1800< n \leq 2500$ &       0.85 &       0.06 &      6.8\% \\
\hline
           & $500\leq n \leq 1100$ &       6.14 &       0.61 &      9.9\% \\

      Gini & $1100< n \leq 1800$ &      12.11 &       1.13 &      9.3\% \\

           & $1800< n \leq 2500$ &       18.4 &       1.69 &      9.2\% \\
\hline
           & $500\leq n \leq 1100$ &       20.2 &       0.66 &      3.3\% \\

  DER(0.5) & $1100< n \leq 1800$ &         63 &       1.42 &      2.3\% \\

           & $1800< n \leq 2500$ &        128 &       2.49 &      1.9\% \\
\hline
\end{tabular}  

\footnotesize{
Note: own calculations using 50 random draws with \texttt{STATA} 14.1 MP (64-bit), 16 GB of RAM, an Intel Core i7 processor and Windows 10 operating system.
}

\end{table}

As expected, for all sample sizes, the fastest RSC to compute is for the variance, while the slowest is the DER(0.5) index, since it involves a non-parametric estimate of a density. The time required to compute the complete RSC increases markedly with the sample size; however, the estimate based on the spline RSC(sp) increases only slightly. For example, for samples between 500 and 1100 observations, computing the RSC of the variance using splines represents 19.7\% of the time it takes with the complete sample, while with larger sample sizes this percentage represents just under 10\%. This saving in computational time is similar for the Gini index (9.5\% average) and definitely more noticeable for the DER (2.5\% average). Figure \ref{fig:figure-times} clearly shows the relative computational advantage of using spline interpolation as larger samples are used.

\begin{figure}
\caption{Comparison of computing times}
\label{fig:figure-times}
\centering
\includegraphics[scale=0.22]{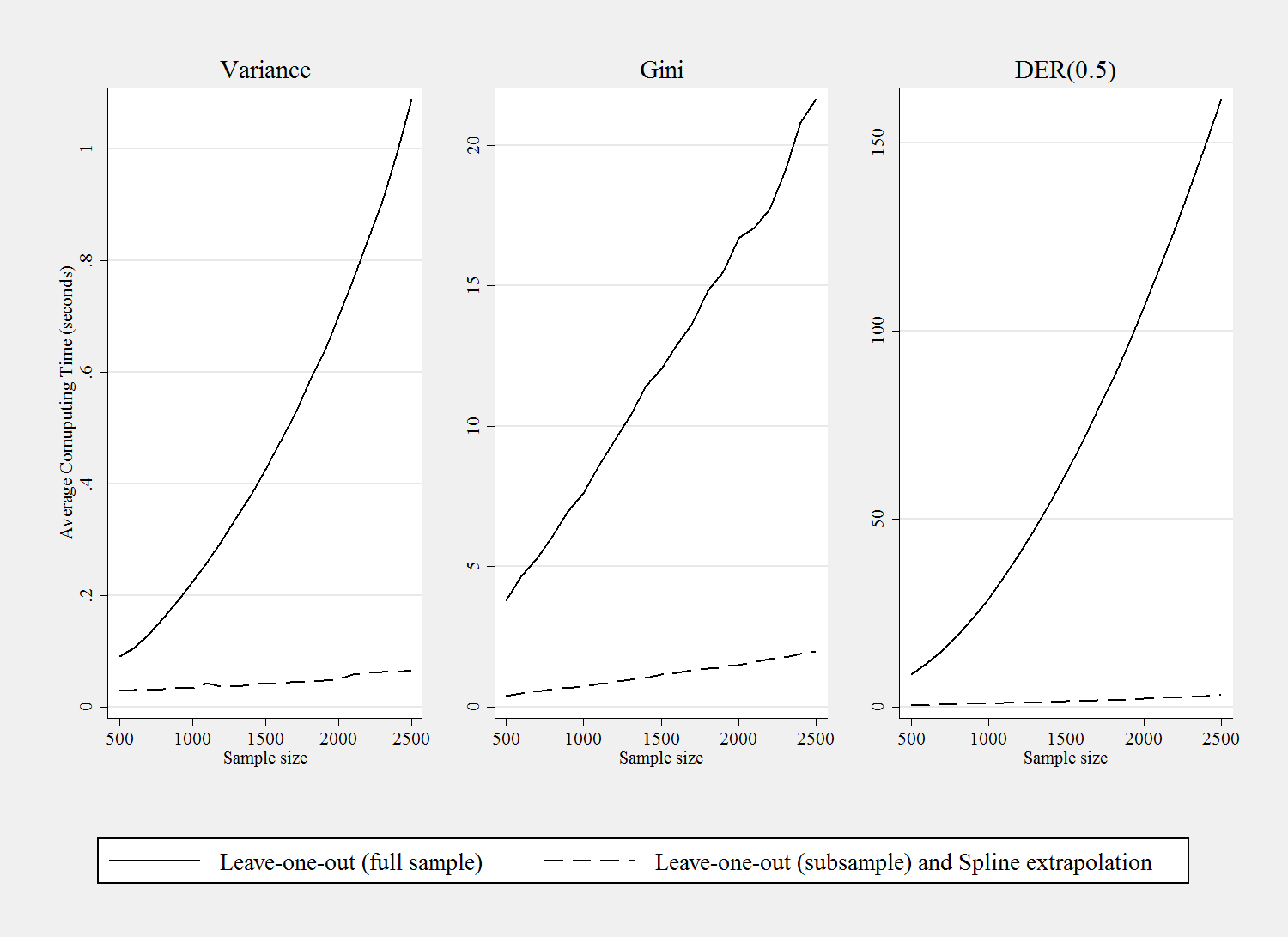}

\footnotesize{
Note: own calculations using 50 random draws with \texttt{STATA} 14.1 MP (64-bit), 16 GB of RAM, an Intel Core i7 processor and Windows 10 operating system.
}

\end{figure}

\section{Empirical illustration} \label{Empirical}

This section presents empirical applications. We first compare the empirical performance of RIF and RSC for the variance and the Gini index, for which the IF can be obtained analytically. Then we add the DER polarization index (Duclos, Esteban, and Ray, 2004) where an explicit analytical closed-form solution for IF is not available (see the Appendix A2 for a description of this index).

We use an extract from the Merged Outgoing Rotation Group of the Current Population Survey of 1983, 1984 and 1985 for males only. More details about the data can be found in Lemieux (2006). The variable of interest is $Y$, the hourly wage, and the covariates $X$ are an indicator of whether the individual is unionized, years of education, whether he is married, non-white, his experience. We use a linear specification in all regressions but, given the results of the previous section, in the case of the polarization index we add a more flexible specification that incorporates the squares and non-trivial interactions of all the covariates.

Obtaining the RSC for each observation using the leave-one-out method can be computationally intensive if $n$ is too large since it requires a separate calculation for each observation. Therefore, we also consider computing the RSC by interpolating an estimated spline using 1000 random points in the distribution of $Y$ (this is denoted as RSC(sp)).

\begin{table}
	\caption{Wage Inequality}
		\label{table1}
	\centering
	\scriptsize
\begin{tabular}{rccc|ccc}
\hline
    {\bf } &  \multicolumn{ 3}{c}{{\bf Variance}} &     \multicolumn{ 3}{|c}{{\bf Gini}} \\

    {\bf } &  {\bf RIF} &  {\bf RSC} & {\bf RSC(sp)} &  {\bf RIF} &  {\bf RSC} & {\bf RSC(sp)} \\
\hline
           &            &            &            &            &            &            \\

     Union &  -15.51*** &  -15.51*** &  -16.42*** &   -6.62*** &   -6.62*** &   -6.89*** \\

           &    (0.189) &    (0.162) &    (0.157) &    (0.055) &    (0.046) &    (0.046) \\

 Education &    1.48*** &    1.48*** &    1.22*** &   -0.49*** &   -0.49*** &   -0.56*** \\

           &    (0.030) &    (0.037) &    (0.035) &    (0.009) &    (0.010) &    (0.010) \\

Experience &    0.15*** &    0.15*** &    0.10*** &   -0.10*** &   -0.10*** &   -0.11*** \\

           &    (0.007) &    (0.008) &    (0.008) &    (0.002) &    (0.002) &    (0.002) \\

   Married &  -10.13*** &  -10.13*** &  -10.98*** &   -5.24*** &   -5.24*** &   -5.44*** \\

           &    (0.190) &    (0.185) &    (0.181) &    (0.056) &    (0.057) &    (0.057) \\

 Non-white &    0.91*** &    0.91*** &    1.33*** &    1.53*** &    1.53*** &    1.65*** \\

           &    (0.262) &    (0.248) &    (0.246) &    (0.077) &    (0.079) &    (0.080) \\

           &            &            &            &            &            &            \\
\hline
Observations &    266,956 &    266,956 &    266,956 &    266,956 &    266,953 &    266,956 \\
\hline
\end{tabular}

\emph{Source: }Extract from the Merged Outgoing Rotation Group of the Current Population Survey of 1983, 1984 and 1985. Notes: Standard errors in parentheses; *** $p<0.01$, ** $p<0.05$, * $p<0.1$; (sp) indicates that the RSC was estimated using a cubic spline with a random subsample of 1000 points; all estimates are multiplied by 100.
	
\end{table}

\begin{table}
	\caption{Wage Polarization: DER($\alpha=0.50$)}
		\label{table2}
	\centering
	\scriptsize
\begin{tabular}{rcc|cc}
\hline
    {\bf } & \multicolumn{ 2}{c}{{\bf RSC}} & \multicolumn{ 2}{|c}{{\bf RSC(sp)}} \\

    {\bf } & {\bf Linear} & {\bf Quad. Form.} & {\bf Linear} & {\bf Quad. Form.} \\
\hline
           &            &            &            &            \\

     Union &   -3.38*** &   -3.43*** &   -3.43*** &   -3.45*** \\

           &    (0.027) &    (0.030) &    (0.025) &    (0.027) \\

 Education &   -0.35*** &   -0.35*** &   -0.40*** &   -0.39*** \\

           &    (0.006) &    (0.006) &    (0.005) &    (0.005) \\

Experience &   -0.07*** &   -0.06*** &   -0.08*** &   -0.06*** \\

           &    (0.001) &    (0.001) &    (0.001) &    (0.001) \\

   Married &   -3.28*** &   -2.42*** &   -3.10*** &   -2.36*** \\

           &    (0.034) &    (0.035) &    (0.031) &    (0.033) \\

 Non-white &    0.99*** &    1.04*** &    0.96*** &    1.01*** \\

           &    (0.050) &    (0.049) &    (0.044) &    (0.043) \\

           &            &            &            &            \\
\hline
Observations &    266,956 &    266,956 &    266,956 &    266,956 \\
\hline
\end{tabular}

\emph{Source: }Extract from the Merged Outgoing Rotation Group of the Current Population Survey of 1983, 1984 and 1985. Notes: Standard errors in parentheses; *** $p<0.01$, ** $p<0.05$, * $p<0.1$; (sp) indicates that the RSC was estimated using a cubic spline with a random subsample of 1000 points; all estimates are multiplied by 100.
	
\end{table}

Table \ref{table1} shows results for the variance and the Gini index. Remarkably, the differences between the RIF and RSC regressions are negligible. Interestingly, the approximation obtained through the spline intrapolation seems to be accurate, suggesting that it is a convenient computational strategy relative to the leave-one-out method.

Table \ref{table2} also shows results for the DER polarization indexes, for which the Gini columns correspond to a particular case ($\alpha=0$), and for proper polarization we set $\alpha=0.5$ following Duclos et al. (2004). We stress the fact that the IF function is not available for this case, hence we obtain results based on the RSC solely.  Note that in this case the coefficients of the linear model give similar results to the average partial effects of the more flexible model. Considering the results of the simulations in the previous section, this is probably due to the large sample size since the RSC approximation to the RIF is more precise. Again, the computationally convenient spline approximation produces similar results than when RSC is computed directly. Even though a detailed study of the effects on inequality and polarization exceeds the scope of this note, we remark that all factors reduce both measures (i.e., higher levels education predict less \emph{unconditionally} inequality and polarization), and that effects are stronger for inequality.

\section*{References}

\noindent Cowell, F.A., Flachaire, E.  2015. Statistical Methods for Distributional Analysis. In Anthony B. Atkinson and
Francois Bourguignon (eds.), Handbook of Income Distribution. Amsterdam: Elsevier.

\vspace{0.4cm}

\noindent Davies, J.B., Fortin, N.M., Lemieux, T. 2017. Wealth inequality: Theory, Measurement and Decomposition. \textit{Canadian Journal of Economics/Revue Canadienne d'\'Economique} 50(5): 1224-1261.

\vspace{0.4cm}

\noindent DiNardo, J., Fortin, N.M., Lemieux, T. 2017. Labor Market Institutions and the Distribution of Wages, 1973-1992: A Semiparametric Approach.
\textit{Econometrica} 64(5): 1001-1044.

\vspace{0.4cm}

\noindent Duclos, J.-Y., Esteban, J., Ray, D. 2004. Polarization: Concepts, Measurement, Estimation. \textit{Econometrica} 72(6): 1737-1772.

\vspace{0.4cm}

\noindent Durrleman, S., Simon, R. 1989. Flexible Regression Models with Cubic Splines. \textit{Statistics in Medicine} 8(5): 551-561.

\vspace{0.4cm}

\noindent Essama-Nssah, B., Lambert, P.J. 2015. Chapter 6: Influence Functions for Policy Impact Analysis. In John A. Bishop and Rafael Salas (eds.), Inequality, Mobility and Segregation: Essays in Honor of Jacques Silber, pp.135-159. Bigley, UK: Emerald Group Publishing Limited.

\vspace{0.4cm}

\noindent Firpo, S.P., Fortin, N.M., Lemieux, T. 2009. Unconditional Quantile
Regressions. \textit{Econometrica} 77(3): 953-973.

\vspace{0.4cm}

\noindent Firpo, S.P., Fortin, N.M., Lemieux, T. 2018. Decomposing Wage Distributions Using Recentered Influence Function
Regressions. \textit{Econometrics} 6(3): 41.

\vspace{0.4cm}

\noindent Fortin, N.M., Lemieux, T., Firpo, S.P. 2011. Decomposition Methods in Economics. In Orley Ashenfelter and David Card (eds.), Handbook of Labor Economics. Amsterdam: Elsevier.

\vspace{0.4cm} 

\noindent Gasparini, L., Horenstein, M., Molina, E., Olivieri, S. 2008. Income Polarization in Latin America: Patterns and Links with Institutions and Conflict, \textit{Oxford Development Studies}, 36: 461-484.

\vspace{0.4cm}

\noindent Hampel, F. 1974. The Influence Curve and its Role in Robust Estimation. \textit{Journal of the American Statistical Association} 69(346): 383-393.

\vspace{0.4cm}

\noindent  Harrell, F. E., Jr. 2001. Regression Modeling Strategies: With Applications to Linear Models, Logistic Regression, and Survival Analysis. New York: Springer.

\vspace{0.4cm}

\noindent Huber, P., Ronchetti, E.M. 2009. \textit{Robust Statistics} (2nd edition). Wiley.

\vspace{0.4cm}
\noindent Jaeckel, L.A. 1972. Estimating regression coefficients by minimizing the dispersion of the residuals. \textit{Annals of Mathematical Statistics} 43, 1449-1458.

\vspace{0.4cm}

\noindent Lemieux, T. 2006. Increasing Residual Wage Inequality: Composition Effects, Noisy Data, or Rising Demand for Skill? \textit{American Economic Review} 96(3): 461-498.

\vspace{0.4cm}

\noindent Nasser, M., Alam, M. 2006. Estimators of Influence Function. \textit{Communications in Statistics - Theory and Methods}, 35(1), 21-32.

\vspace{0.4cm}

\noindent Newson, R. B. 2012. Sensible parameters for univariate and multivariate splines. \textit{Stata Journal}, 12: 479–504

\vspace{0.4cm}

\noindent Orsini, N., and S. Greenland. 2011. A procedure to tabulate and plot results after flexible modeling of a quantitative covariate. \textit{Stata Journal}, 11, 1–29.

\vspace{0.4cm}

\noindent Smith, P. L. 1979. Splines as a useful and convenient statistical tool. \textit{American Statistician}, 33, 57–62.

\vspace{0.4cm}
\noindent Tukey, J.W. 1977. Exploratory Data Analysis, Addison-Wesley, Reading, MA.

\vspace{0.4cm}

\noindent von Mises, R. 1947. On the Asymptotic Distribution of Differentiable Statistical Functions. \textit{Annals of Mathematical Statistics} 18(3): 309-348.

\vspace{0.4cm}

\noindent Wegman, E. J., and I. W. Wright. 1983. Splines in statistics. \textit{Journal of the American Statistical Association}, 78, 351–365.

\newpage

\section*{Appendix A1}
\textbf{Proof of Proposition 1.}

Using eq. \eqref{v} with $F_n$ and $F_n^{(j)}$ for the case of $t=1$: 
\begin{equation}
 v\left(F_{n}\right)=v\left(F_{n}^{(j)}\right)+\int \psi_{n}(y) d\left(F_{n}-F_{n}^{(j)}\right)(y)+r\left(\tilde{t}, F_{n}, F_{n}^{(j)}\right),
\end{equation}
for some $\tilde{t}\in[0,1]$. Note that $\psi_{n}(y)=I F\left(y, v, F_{n}\right) \stackrel{p}{\rightarrow} \psi(y)$ by continuity of the probability limit.

Now note that $n\left[F_{n}-F_{n}^{(j)}\right]=1\left(y_{j}<y\right)+O_{p}(1)$ because

\begin{equation*}
\begin{array}{c}{F_{n}(y)=\frac{1}{n} 1\left(y_{j} \leq y\right)+\frac{n-1}{n} F_{n}^{(j)}(y)}, \\ {F_{n}(y)-F_{n}^{(j)}(y)=\frac{1}{n} 1\left(y_{j} \leq y\right)+\frac{n-1}{n} F_{n}^{(j)}(y)-F_{n}^{(j)}(y)}, \\ {F_{n}(y)-F_{n}^{(j)}(y)=\frac{1}{n} 1\left(y_{j} \leq y\right)-\frac{1}{n} F_{n}^{(j)}(y)}.\end{array}
\end{equation*}
That is,

\begin{equation}\label{nF}
n\left[F_{n}(y)-F_{n}^{(j)}\right]=1\left(y_{j} \leq y\right)-a_{n},
\end{equation}
with $a_{n}=F_{n}^{(j)}(y) \stackrel{p}{\rightarrow} F(y)$ by the Law of Large Numbers.

Then,

\begin{equation}\label{nv}
\begin{array}{c}{n \cdot\left[v\left(F_{n}\right)-v\left(F_{n}^{(j)}\right)\right]=\int \psi_{n}(y) d\left(1\left(y_{j} \leq y\right)-a_{n}\right)(y)+n \cdot r\left(\tilde{t}, F_{n}, F_{n}^{(j)}\right)} \\ {n \cdot\left[v\left(F_{n}\right)-v\left(F_{n}^{(j)}\right)\right]=\int \psi_{n}(y) d\left(1\left(y_{j} \leq y\right)\right)(y)-\int \psi_{n}(y) d\left(a_{n}\right)(y)+n \cdot r\left(\tilde{t}, F_{n}, F_{n}^{(j)}\right)}\end{array}
\end{equation}

Using the fact that $1\left(y_{j} \leq y\right)$ is the Dirac function, the first term of eq. \eqref{nv} is

\begin{equation*}
\int \psi_{n}(y) d\left(1\left(y_{j} \leq y\right)\right)(y)=\psi_{n}\left(y_{j}\right) \stackrel{p}{\rightarrow} \psi\left(y_{j}\right).
\end{equation*}

Noting that $a_{n} \stackrel{p}{\rightarrow} F(y)$ and $\psi_{n}(y) \stackrel{p}{\rightarrow} \psi(y)$, by continuity of the probability limit, the second term of \eqref{nv} becomes 
\begin{equation*}
\operatorname{plim} \int \psi_{n}(y) d\left(a_{n}\right)(y)=\int \psi(y) d F(y)=0,
\end{equation*}
because of property (i). Then,
\begin{equation*}
\operatorname{plim} \int \psi(y) d\left(1\left(y_{j} \leq y\right)+a_{n}\right)(y)=\int \psi(y) d\left(1\left(y_{j} \leq y\right)\right)(y)=\psi(y).
\end{equation*}

It remains to study the third term in \eqref{nv}. From \eqref{r},

\begin{equation*}
\begin{array}{c}{n \cdot r\left(\tilde{t}, F_{n}, F_{n}^{(j)}\right)=n \cdot \frac{\tilde{t}^{2}}{2} \iint \psi(y, z) d\left(F_{n}-F_{n}^{(j)}\right)(y) d\left(F_{n}-F_{n}^{(j)}\right)(z)} \\ {n \cdot r\left(\tilde{t}, F_{n}, F_{n}^{(j)}\right)=\frac{1}{n} \cdot \frac{\tilde{t}^{2}}{2} \iint \psi(y, z) d\left[n\left(F_{n}-F_{n}^{(j)}\right)\right](y) d\left[n\left(F_{n}-F_{n}^{(j)}\right)\right](z)}\end{array}
\end{equation*}
for some $\tilde{t}\in[0,1]$. Then using \eqref{nF} and property (ii),

\begin{equation*}
\operatorname{plim} \iint \phi(y, z) d\left[n\left(F_{n}-F_{n}^{(j)}\right)\right](y) d\left[n\left(F_{n}-F_{n}^{(j)}\right)\right](z)=\phi\left(y_{j}, y_{j}\right).
\end{equation*}
Then it follows that
\begin{equation*}
\operatorname{plim}\left\{n \cdot r\left(\tilde{t}, F_{n}, F_{n}^{(j)}\right)\right\}=\left(\operatorname{plim} \frac{1}{n}\right) \cdot \frac{\tilde{t}^{2}}{2} \phi\left(y_{j}, y_{j}\right)=0.
\end{equation*}
Then, the result follows,
\begin{equation*}
\operatorname{plim} \frac{v\left(F_{n}\right)-v\left(F_{n}^{(j)}\right)}{1 / n}=\psi\left(y_{j}\right)=I F\left(y_{j}, v, F\right). \ QED
\end{equation*}

\section*{Appendix A2: Polarization index}

We motivate the case of a model where the IF is not available: the DER polarization index (Duclos, Esteban, and Ray, 2004). 

Polarization is an important welfare concept in economics and political science. Intuitively, it measures the tension between individuals in a society, that depends positively on how distant individuals are between groups (alienation) and how close they are within a group (identification). From this perspective, a standard measure of inequality like the Gini index focuses on just the first component. Duclos et al. (2004) provide a full axiomatic framework that leads to a logically coherent measure of polarization. For a detailed empirical study on polarization for the case of Latin America and the Caribbean, see Gasparini et al. (2008).

Let $y_1, y_2, \ldots, y_n$ be and iid sample of incomes, ordered from lowest to highest. Duclos et al. (2004) propose the following empirical measure of polarization:

\[ P_\alpha = \frac 1 n \sum_{i=1}^n \hat f(y_i)^{\alpha} \hat a (y_i)\]
where $\hat a(y_i)  = \hat \mu + y_i\left(n^{-1} (2i-1) - 1\right) -  n^{-1} \left(2 \sum_{j=1}^{i-1} y_j + y_i \right)$, $\hat \mu$ is the sample mean and $\hat f (y_i)$ is an estimate of the  density of incomes. The parameter $\alpha$ is set exogenously and plays a key role in characterizing polarization. As a matter of fact, when $\alpha=0$ polarization reduces to the Gini index (note that for this particular case the IF is available). Larger values of $\alpha$ result in the index giving relatively more importance to identification, that is, to how close individuals are `surrounded' by others of similar income. The axiomatic approach of Duclos et al. (2004) imposes lower and upper bounds to the values $\alpha$ may take in practice.

\pagebreak
\newpage

\end{document}